\documentclass[12pt]{article}
\usepackage[T1]{fontenc}
\usepackage[latin9]{inputenc}
\usepackage{float}
\usepackage{amsmath}
\usepackage{amssymb}
\usepackage{graphicx}
\usepackage{esint}

\makeatletter
\@ifundefined{date}{}{\date{}}

\usepackage{epsfig}
\usepackage{euscript}
\usepackage{latexsym}
\usepackage{amsthm}
\usepackage{color}
\usepackage{multirow}
\usepackage{epsfig}

\theoremstyle{remark}

\theoremstyle{plain}
\newtheorem{theorem}{Theorem}

\newtheorem{definition}{Definition}
\newtheorem{remark}{Remark}

\makeatother

\begin{document}

\title{\vspace{1cm}
 \textbf{Application of $p$-adic analysis methods in describing Markov processes on ultrametric spaces isometrically embeddable into $\mathbb{Q}_{p}$} }

\author{A.\,Kh.~Bikulov \medskip{}
\\
\textit{Institute of Chemical Physics, }
\\
\textit{Kosygina Street 4, 117734 Moscow, Russia} \medskip{}
\\
e-mail:\:\texttt{beecul@mail.ru}
\bigskip{}
 \\
 and
 \bigskip{}
 \\
 A.\,P.~Zubarev
 \medskip{}
 \\
\textit{ Physics Department, Samara State Aerospace University,  }
\\
\textit{ Moskovskoe shosse 34, 443123, Samara, Russia } \medskip{}
 \\
\textit{Physics and Chemistry Department, }
 \\
 \textit{Samara State University of Railway Transport,}
 \\
\textit{Perviy Bezimyaniy pereulok 18, 443066, Samara, Russia }\medskip{}
 \\
e-mail:\:\texttt{apzubarev@mail.ru}
}
\maketitle

\begin{abstract}
We propose a method for describing  stationary Markov processes on the class of
ultrametric spaces~$\mathbb{U}$ isometrically embeddable  in the field  $\mathbb{Q}_{p}$
of  $p$-adic numbers. This method is capable of reducing the study of such processes to the investigation of processes on~$\mathbb{Q}_{p}$.
Thereby the traditional machinery of $p$-adic mathematical physics can be applied to calculate the
characteristics of stationary Markov processes on such spaces.
The Cauchy problem for
the Kolmogorov--Feller equation of a~stationary Markov process on such spaces is shown as being reducible to
the Cauchy problem for a~pseudo-differential equation
on~$\mathbb{Q}_{p}$ with  non-translation-invariant measure $m\left(x\right)d_{p}x$.
The spectrum of the pseudo-differential operator of the Kolmogorov--Feller equation
on $\mathbb{Q}_{p}$ with measure $m\left(x\right)d_{p}x$ is found. Orthonormal basis of real valued functions for  $L^{2}\left(\mathbb{Q}_{p},m\left(x\right)d_{p}x\right)$
is constructed from the eigenfunctions of this operator.
\end{abstract}

\section{Introduction}
The last three decades have seen a~growing interest
to ultrametric models in various branches of physics, biology, economics, and social sciences:  spin glass models, biopolymer models, fractal
structures, optimization theory, taxonomy, evolutionary biology, cluster and factor analysis, etc.\ (see the surveys \cite{RTV},~\cite{DKKV}).
Dynamic probability models on ultrametric structures, and in particular, random walk models have been studied on general ultrametric spaces \cite{OS}, \cite{BH}, \cite{BH1}, \cite{Motyl},
\cite{KK},~\cite{Koz2}, on the field of $p$-adic numbers
\cite{V}, \cite{VVZ}, \cite{ABK}, \cite{ABKO}, \cite{ABO}, \cite{ABZ}, \cite{ABZ_motor}, \cite{ABZ_MIAN},  and on the ring of $m$-adic  numbers~\cite{DZ}.
Such models were shown as being directly related to the description of the dynamics of conformational rearrangements of protein molecules.
There is also a~ground to believe that such ultrametric models may be related to the description of relaxation processes
in complex socio--economic systems \cite{MS}, \cite{SJ},~\cite{BZK}.

Ultrametric modeling of similar systems resides in the description of the system dynamics as a random process on the ultrametric space
of system configurations, when the probability transitions  between configurations are defined by an ultrametric distance between them.
From the mathematical point of view, such a~dynamics is described by one equation or a~system of
 ``reaction--ultrametric diffusion''-type equations. Similar equations admitting exact analytic solutions on $p$-adic spaces
were considered in the papers \cite{VVZ}, \cite{ABK}, \cite{ABKO}, \cite{ABO} as
models of ultrametric random walk on protein energy landscapes,
models of the ligand rebinding kinetics of myoglobin, and models of prototypes of molecular nanomachines. A~``reaction--ultrametric diffusion''-type
equation on  the set of $p$-adic numbers~$\mathbb{Q}_{p}$ is as follows
\begin{equation}\label{p-RUD}
\frac{\partial f(x,t)}{dt}=-D^{\alpha}f\left(x,t\right)+\lambda V\left(x\right)f(x,t),
\end{equation}
where $D^{\alpha}$ is the Vladimirov pseudo-differential operator~\cite{V}, \cite{VVZ} and
$V\left(x\right)$ is some function on~$\mathbb{Q}_{p}$. In the papers
\cite{OS}, \cite{BH}, \cite{BH1}, the equations of ultrametric random
walk on general finite ultrametric spaces were studied and exact analytic solutions of a~Cauchy problem for a~certain class of initial conditions were found.
In~\cite{KK}, \cite{Koz2}, a~wavelet basis on a~general regular ultrametric space~$\mathbb{U}$ was described,
a~measure $\mu$ was defined for a~compact~$\mathbb{U}$, and the spectrum was found for the analogue of
the Vladimirov operator on~$\mathbb{U}$,
\begin{equation}
Tf\left(v\right)=\intop_{\mathbb{U}}d\mu_{\mathbb{U}}\left(u\right)\dfrac{1}{\left(d\left(u,v\right)\right)^{\alpha+1}}\left(f\left(v\right)-f\left(u\right)\right),\label{T}
\end{equation}
which acts on functions from $L^{2}\left(\mathbb{U},d\mu_{\mathbb{U}}\right)$.

The present paper is mainly concerned with the study of a pseudo-differential operator $W_{m\left(x\right)}$ of the form
\begin{equation}
W_{m\left(x\right)}f\left(x\right)=\intop_{\mathbb{Q}_{p}}d_{p}y
m\left(y\right)W\left(|x-y|_{p}\right)\left(f\left(y\right)-f\left(x\right)\right),\label{W}
\end{equation}
where $m\left(x\right)$ is some nonnegative function on~$\mathbb{Q}_{p}$ which is locally integrable with respect to the Haar measure $d_{p}x$.

There are at least two compelling reasons for such statement of the problem.

The first one stems from the need for the study of Markov random process
on inhomogeneous ultrametric spaces. On such spaces, for each ball the number of maximal subballs may be different. Inhomogeneous ultrametric spaces
correspond to boundary points of ultrametric trees with different branch numbers at each vertex.
Of even greater interest, however, is the class of problems involving random processes on ultrametric trees with random branching.
Since the problem of analytic description of general inhomogeneous ultrametric spaces is quite difficult to formalize,
it is quite natural to reduce the equations describing random processes on an arbitrary ultrametric space~$\mathbb{U}$ to equations
in the standard space~$\mathbb{Q}_{p}$. Such a~problem can be easily formalized for the so-called
ultrametric spaces that can be isometrically embedded into $\mathbb{Q}_{p}$; that is, for spaces~$\mathbb{U}$ which are
isometrically isomorphic to some subset  $M\subset\mathbb{Q}_{p}$ of nonzero measure. For such spaces one may define a~measure,
which is the natural restriction of the Haar measure on~$\mathbb{Q}_{p}$ to $M\subset\mathbb{Q}_{p}$.
Consequently, the description of a~Markov random process on
an ultrametric space~$\mathbb{U}$ isometrically embeddable into~$\mathbb{Q}_{p}$
is equivalent to the description of a~Markov random process on $M\subset\mathbb{Q}_{p}$.
For random processes on $M\subset\mathbb{Q}_{p}$, the solution $f\left(x,t\right)$
of the Cauchy problem for the Kolmogorov--Feller equation (or the master equation)~\cite{G}
\begin{equation}\label{K-F}
\frac{d\varphi(x,t)}{dt}=\intop_{M}d_{p}y
W\left(|x-y|_{p}\right)\left(\varphi\left(y,t\right)-\varphi\left(x,t\right)\right)
\end{equation}
for a stationary Markov process  on~$M$ can be recovered as
\begin{equation}
\varphi\left(x,t\right)=m\left(x\right)f\left(x,t\right),\label{sol_f}
\end{equation}
where $m\left(x\right)$ is the characteristic function of a~subset
$M\subset\mathbb{Q}_{p}$, provided that one knows the solution $f\left(x,t\right)$
of the Cauchy problem for the equation
\begin{equation}\label{V_Eq}
\dfrac{df\left(x,t\right)}{dt}=W_{m\left(x\right)}f\left(x,t\right)
\end{equation}
with operator~(\ref{W}) on~$\mathbb{Q}_{p}$.
It is worth pointing out that $\mathbb{Q}_{p}=M\cup\left(\mathbb{Q}_{p}\setminus M\right)$,
and hence the solution $f\left(x,t\right)$ of equation~(\ref{V_Eq}) can be written as the sum of two functions
$f\left(x,t\right)=\varphi\left(x,t\right)+\phi\left(x,t\right)$, whose supports lie, respectively, in~$M$ and $\mathbb{Q}_{p}\setminus M$.
The function $\phi\left(x\right)$, which gives no contribution to the solution~(\ref{sol_f}),
will always be present in the solution $f\left(x,t\right)$.
This seems to be the ``price'' paid for the possibility of reduction of the Cauchy problem of equation~(\ref{K-F}) on~$\mathbb{U}$ to the analogous
equation (\ref{V_Eq}) on~$\mathbb{Q}_{p}$, which can be solved by the standard machinery of
the $p$-adic mathematical physics.

The second reason for the study of operator~(\ref{W}) is dictated by the need to study the solutions of the equation of ultrametric random walk on
$\mathbb{Q}_{p}$ with potential $U\left(y\right)$. The most natural form of such an equation is as follows:
\begin{equation}
\frac{\partial f(x,t)}{dt}=\intop_{\mathbb{Q}_{p}}d_{p}yW\left(|x-y|_{p}\right)\left(U\left(y\right)f\left(y,t\right)-U\left(x\right)f\left(x,t\right)\right).\label{UMD_U}
\end{equation}
So far, there is no known exact analytic solution of equation~(\ref{UMD_U}) with some nontrivial potential $U\left(y\right)$.
Nonetheless, it is of great value that equation~(\ref{UMD_U}) can be reduced to a~``reaction--ultrametric diffusion''-type equation
with operator~(\ref{W}) and measure $m\left(x\right)d_{p}x=U\left(x\right)d_{p}x$,
\begin{equation}
\frac{\partial f(x,t)}{dt}=W_{U\left(x\right)}f\left(x,t\right)+V\left(x\right)f\left(x,t\right),\label{UMD_m}
\end{equation}
where
\[
V\left(x\right)=\intop_{\mathbb{Q}_{p}}d_{p}yW\left(|x-y|_{p}\right)\left(U\left(y\right)-U\left(x\right)\right).
\]
Thus, studying the properties and the spectrum of operator~(\ref{W}) with various measures $m\left(x\right)d_{p}x$
provides the possibility of searching a~method for solution of equations of type~(\ref{UMD_U}).

The present paper is organized as follows. In Section~2 we define the class of regular
ultrametric spaces isometrically embeddable into $\mathbb{Q}_{p}$  and reduce the problem of solving the Kolmogorov--Feller equation~(\ref{K-F}) for a~stationary Markov process
to the problem of solving equation~(\ref{V_Eq}). In Section~3 we find the spectrum of operator~(\ref{W})
in the space $L^{2}\left(\mathbb{Q}_{p},m\left(x\right)d_{p}x\right)$.
The results obtained in this section are the natural generalizations of the results of \cite{BH},~\cite{BH1}. It is also worth mentioning
that there are differences between our results and the results of \cite{KK}, \cite{Koz2} for the spectrum of an~analogue of the Vladimirov operator on~$\mathbb{U}$.
This is explained by the substantial difference in the definition of measure $\mu_{\mathbb{U}}\left(u\right)$ on~$\mathbb{U}$. In Section~4
we explicitly construct a~orthonormal basis of real valued functions for  $L^{2}\left(\mathbb{Q}_{p},m\left(x\right)d_{p}x\right)$
from the eigenfunctions of operator~(\ref{W}).

\section{Ultrametric spaces isometrically embeddable into $\mathbb{Q}_{p}$}

Let $\mathbb{U}$ be an arbitrary ultrametric space of elements $\left\{ u\right\} $ with ultrametric $\delta\left(u,u'\right)$
$\forall u,u'\in\mathbb{U}$. By a~ball in~$\mathbb{U}$ of radius~$r$ and centre
at a~point $v\in\mathbb{U}$ we shall mean the set $\mathcal{B}_{r}\left(v\right)=\left\{ u\in\mathbb{U}:\:\delta\left(u,v\right)\leq r\right\} $.

\begin{definition} Following~\cite{Koz2}, an ultrametric space~$\mathbb{U}$ is called \textit{regular} if it is complete and has the following properties:

1) the set of all balls in $\mathbb{U}$ of nonzero radius is at most countable;

2) for any infinite nested sequence of balls
$\left\{ \mathcal{B}_{r_{i}}\right\} $, $\mathcal{B}_{r_{i+1}}\subset\mathcal{B}_{r_{i}}$
the radii $r_{i}$ of the balls tend to zero as $i\rightarrow\infty$;

3) each ball of nonzero radius is a~finite union of maximal subballs.
\label{def1}
\end{definition}
A wavelet basis $\left\{ \psi_{Ij}\right\} $ on a~regular
ultrametric space~$\mathbb{U}$ was described in~\cite{KK}, \cite{Koz2}. It was also shown that
$\psi_{Ij}$ is an eigenvector of an analogue of the Vladimirov operator~(\ref{T}) on~$\mathbb{U}$ with eigenvalues
\begin{equation}
\lambda_{I}=\nu^{-\alpha}\left(I\right)+\sum_{J>I}\nu^{-\alpha}\left(J\right)\left(1-p_{J}^{-1}\right),\label{lambda(I)}
\end{equation}
where $\nu\left(I\right)$ is the measure of a~ball~$I$ and $p_{J}$  is the number of maximal subballs in a~ball~$J$. In~(\ref{lambda(I)}) index~$I$ runs over
all possible balls in~$\mathbb{U}$, and $J$~varies over all orthogonal functions from the space $V_{0}\left(I\right)\subset V(I)$. Note that
$V(I)$ is the space of functions on~$\mathbb{U}$, this space is generated by the characteristic
functions of maximal subbals in a~nonzero ball~$I$. Also, $V_{0}\left(I\right)$ is the subspace
of~$V(I)$ consisting of functions with zero mean with respect to the measure $d\mu_{\mathbb{U}}\left(u\right)$.
The measure $\mu_{\mathbb{U}}\left(u\right)$ on~$\mathbb{U}$ is generated by the $\sigma$-algebra of all possible balls in~$\mathbb{U}$.
In~\cite{KK}, \cite{Koz2}, the measure of each ball was defined to be equal to its radius. With such definition of a~measure all maximal subballs
of any ball will have the same ``volume'', despite the fact that they, in turn, may contain a~different number of subballs of the same
``volume''. From our point of view, such definition of a~measure on an arbitrary
regular ultrametric space is not very convenient for physical applications, but as yet no other method for constructing a~measure on
arbitrary regular ultrametric spaces is available. Nevertheless, there is a~fairly broad class of ultrametric spaces, which we call ultrametric spaces isometrically embeddable
into $\mathbb{Q}_{p}$. A~measure on such a~space is the natural restriction of the Haar measure on~$\mathbb{Q}_{p}$.

Let $\mathbb{Q}_{p}$ be the field of $p$-addic numbers and $M$ be some measurable
subset of~$\mathbb{Q}_{p}$. Then $M$~is an ultrametric space
with ultrametric  induced by the ultrametric on~$\mathbb{Q}_{p}$
and with measure $d\mu_{M}\left(x\right)$, which is the restriction of the Haar measure
$d_{p}x$ on~$\mathbb{Q}_{p}$. We let $m\left(x\right)$ denote
the characteristic  function of a~subset~$M$. Then $m\left(x\right)$
is measurable and $d\mu_{M}\left(x\right)=m\left(x\right)d_{p}x$.

Assume that $\mathbb{U}$ is regular. Then the possible values $\delta_{i}$
of the ultrametric $\delta\left(u,u'\right)$ form a~countable increasing
sequence $\left\{ \delta_{i}\right\} $, $\delta_{i}<\delta_{i+1}$
$i\in\mathbb{Z}$. Besides, for any prime number~$p$, the function
\begin{equation}
\label{gen_umetr}
d\left(u,u'\right)=p^{i}\quad \text{for}\ \ \delta\left(u,u'\right)=\delta_{i}
\end{equation}
is also an ultrametric on the space~$\mathbb{U}$, which is equivalent to the ultrametric  $\delta\left(u,u'\right)$.
In what follows, the set $\mathcal{B}_{i}\left(v\right)=\left\{ u\in\mathbb{U}:\: d\left(u,v\right)\leq p^{i}\right\} $
will be referred to as a~ball in~$\mathbb{U}$ of radius $p^{i}$ with centre at a~point~$v$.

\begin{definition} A~regular ultrametric space~$\mathbb{U}$ is \textit{isometrically embeddable into} $\mathbb{Q}_{p}$ if there exists a~subset
$M\subset\mathbb{Q}_{p}$ such that $\mathbb{U}$ is isometrically isomorphic to~$M$.
\label{def2}
\end{definition}

Note that instead of the field of $p$-adic numbers $\mathbb{Q}_{p}$ one may consider the
pseudo-normed ring of $m$-adic numbers $\mathbb{Q}_{m}$ (see~\cite{DZ}). In this case, $p$~is assumed to be an arbitrary natural number~$m$, which is not necessarily prime.

Is is clear that the description of a~Markov random process on an ultrametric space~$\mathbb{U}$ isometrically embeddable into~$\mathbb{Q}_{p}$
is equivalent to the description of a~Markov process on $M\subset\mathbb{Q}_{p}$.
Let us outline the procedure of reduction of the Kolmogorov--Feller equation~(\ref{K-F}) for such a~process to the Kolmogorov--Feller equation with
modified measure and~(\ref{V_Eq}) in~$\mathbb{Q}_{p}$.

Consider a stationary Markov process  on~$M$ (an ultrametric random walk) described by equation~(\ref{K-F}).
Here, $W\left(d\left(x,y\right)\right)=\lim_{t'\rightarrow t}\dfrac{d}{dt'}p\left(y,t'|x,t\right)$, where $p\left(y,t'|x,t\right)=p\left(d\left(x,y\right),t'-t\right)$
is the transition function of such a~process. Instead of equation~(\ref{K-F}), we may
consider an equation of form~(\ref{V_Eq}). Since $m\left(x\right)$ is
the characteristic function of a~set~$M$, we may denote
\[
f\left(x,t\right)=\varphi\left(x,t\right)+\phi\left(x,t\right),\;\varphi\left(x,t\right)=m\left(x\right)f\left(x,t\right),\:\phi\left(x,t\right)=\left(1-m\left(x\right)\right)f\left(x,t\right).
\]
Hence, for the functions $\varphi\left(x,t\right)$ and $\phi\left(x,t\right)$, we have
\begin{equation}
\dfrac{d\varphi\left(x,t\right)}{dt}=\intop_{M}d_{p}y W\left(|x-y|_{p}\right)\left(\varphi\left(y,t\right)-\varphi\left(x,t\right)\right)
\,\, {\rm for} \,\,x \in M,
\label{eq_varphi}
\end{equation}
\begin{equation}
\dfrac{d\phi\left(x,t\right)}{dt}=\intop_{M}d_{p}y
W\left(|x-y|_{p}\right)\left(\varphi\left(y,t\right)-\phi\left(x,t\right)\right)
\,\, {\rm for} \,\,x \in \mathbb{Q}_{p} \backslash M.
\label{eq_phi}
\end{equation}
The equation (\ref{eq_varphi}) coincides with equation (\ref{K-F}). Assume that the initial condition
is supported in~ $M$. Then $f\left(x,0\right)=\varphi\left(x,0\right)$,
and hence, equation (\ref{eq_varphi}) is independent of equation~(\ref{eq_phi}) and admits a~unique solution of the Cauchy problem
with initial condition $\varphi\left(x,0\right)$. Equation (\ref{eq_varphi}) is a~bit more difficult for solution than equation~(\ref{V_Eq}).
Hence, instead of finding a~solution $\varphi\left(x,t\right)$ of equation (\ref{eq_varphi}) it is easier to find a~solution $f\left(x,t\right)$ of equation~(\ref{V_Eq}).
In this case,
\begin{equation}
\varphi\left(x,t\right)=m\left(x\right)f\left(x,t\right).\label{sol-f}
\end{equation}
Clearly, $\phi\left(x\right)$ gives no contribution to solution~(\ref{sol-f}), but this part of the function will always be present in the
solution $f\left(x,t\right)$. This is a~particular feature of the reduction of the Cauchy problem of equation~(\ref{K-F}) on~$\mathbb{U}$ to the similar problem
for equation (\ref{V_Eq}) on~$\mathbb{Q}_{p}$ and the solution thereof with application of the standard methods of the $p$-adic mathematical physics.

\section{The spectrum of the operator $W_{m\left(x\right)}$ on $L^{2}\left(\mathbb{Q}_{p},m\left(x\right)d_{p}x\right)$}

We shall be concerned with operator~(\ref{W}) in $L^{2}\left(\mathbb{Q}_{p}\text{,}\: m\left(x\right)d_{p}x\right)$
\begin{equation}
W_{m\left(x\right)}f\left(x\right)=\intop_{\mathbb{Q}_{p}}d_{p}ym\left(y\right)W\left(|x-y|_{p}\right)\left(f\left(y\right)-f\left(x\right)\right).\label{Vlad_Q_p_gen}
\end{equation}
For our purposes described in the introduction we shall assume that
$m\left(x\right)$ is some function on $\mathbb{Q}_{p}$ which is locally integrable with respect to measure
$d_{p}x$, which assumes the values on $\mathbb{R}_{+}$ and not bounded by the values $0$ and~$1$. We also impose the following condition
on $m\left(x\right)$:
\begin{equation}
\exists\beta>1:\;\lim_{i\rightarrow\infty}\dfrac{i^{\beta}}{\intop_{\mathbb{Q}_{p}}m\left(x\right)d_{p}x\varOmega\left(|x|_{p}p^{-i}\right)}=0.\label{restr}
\end{equation}
The meaning of this condition becomes clear below. We set
\[
V_{i}\left(x\right)=\intop_{\mathbb{Q}_{p}}m\left(y\right)d_{p}y\varOmega\left(|y-x|_{p}p^{-i}\right).
\]
By (\ref{restr})
\begin{equation}
\exists\beta>1:\quad \lim_{i\rightarrow\infty}\dfrac{i^{\beta}}{V_{i}\left(x\right)}=0.\label{restr1}
\end{equation}

\begin{theorem}The functions
\begin{equation}
f_{\gamma,n,a}\left(x\right)=\varOmega\left(|x-np^{-\gamma}-ap^{-\gamma}|_{p}p^{-\gamma+1}\right)-
\dfrac{V_{\gamma-1}\left(np^{-\gamma}+ap^{-\gamma}\right)}{V_{\gamma}\left(np^{-\gamma}\right)}\varOmega\left(|x-np^{-\gamma}|_{p}p^{-\gamma}\right)\label{f_gna},
\end{equation}
where $\gamma\in\mathbb{Z}$, $n\in\mathbb{Q}_{p}/Z_{p}$, $a=0,1,\ldots,p-1$   are such that $V_{\gamma-1}\left(np^{-\gamma}+ap^{-\gamma}\right)\neq0$,
are eigenfunctions of operator~\eqref{Vlad_Q_p_gen} with eigenvalues
\begin{equation}
\lambda_{\gamma,n}=-\sum_{i=\gamma}^{\infty}\left(W\left(p^{i}\right)-W\left(p^{i+1}\right)\right)V_{i}\left(np^{-\gamma}\right).\label{lambda_gn}
\end{equation}
\label{th1}
\end{theorem}

\begin{proof}
We find the effect of operator~(\ref{Vlad_Q_p_gen}) on the characteristic
function $\varOmega\left(|x-np^{-\gamma}-ap^{-\gamma}|_{p}p^{-\gamma+1}\right)$
of the ball $B_{\gamma-1}\left(np^{-\gamma}+ap^{-\gamma}\right)$. If $|x-np^{-\gamma}-ap^{-\gamma}|_{p}>p^{\gamma-1}$,
then
\[
W_{m\left(x\right)}\varOmega\left(|x-np^{-\gamma}-ap^{-\gamma}|_{p}p^{-\gamma+1}\right)=
\]
\[
=\intop_{\mathbb{Q}_{p}}m\left(z+x\right)W\left(|z|_{p}\right)\varOmega\left(|z+x-np^{-\gamma}-ap^{-\gamma}|_{p}p^{-\gamma+1}\right)d_{p}z=
\]
\[
=W\left(|x-np^{-\gamma}-ap^{-\gamma}|_{p}\right)\intop_{\mathbb{Q}_{p}}m\left(z\right)\varOmega\left(|z-np^{-\gamma}-ap^{-\gamma}|_{p}p^{-\gamma+1}\right)d_{p}z=
\]
\[
=W\left(|x-np^{-\gamma}-ap^{-\gamma}|_{p}\right)V_{\gamma-1}\left(np^{-\gamma}+ap^{-\gamma}\right)
\]
Next, if $|x-np^{-\gamma}-ap^{-\gamma}|_{p}\leq p^{\gamma}$, then
\[
W_{m\left(x\right)}\varOmega\left(|y-np^{-\gamma}-ap^{-\gamma}|_{p}p^{-\gamma+1}\right)=
\]
\[
=-\intop_{\mathbb{Q}_{p}}m\left(z\right)\left(1-\varOmega\left(|z-np^{-\gamma}-ap^{-\gamma}|_{p}p^{-\gamma+1}\right)\right)W\left(|z-x|_{p}\right)d_{p}z=
\]
\[
=-\intop_{\mathbb{Q}_{p}}m\left(z+np^{-\gamma}+ap^{-\gamma}\right)\left(1-\varOmega\left(|z|_{p}p^{-\gamma+1}\right)\right)W\left(|z|_{p}\right)d_{p}z=
\]
\[
=-\sum_{i=\gamma}^{\infty}W\left(p^{i}\right)\intop_{\mathbb{Q}_{p}}m\left(z+np^{-\gamma}+ap^{-\gamma}\right)\left[\varOmega\left(|z|_{p}p^{-i}\right)-\varOmega\left(|z|_{p}p^{-i+1}\right)\right]d_{p}z=
\]
\[
=-\sum_{i=\gamma}^{\infty}W\left(p^{i}\right)\intop_{\mathbb{Q}_{p}}m\left(z\right)\left[\varOmega\left(|z-np^{-\gamma}-ap^{-\gamma}|_{p}p^{-i}\right)-\varOmega\left(|z-np^{-\gamma}-ap^{-\gamma}|_{p}p^{-i+1}\right)\right]d_{p}z=
\]
\[
=-\left(\sum_{i=\gamma}^{\infty}W\left(p^{i}\right)V_{i}\left(np^{-\gamma}\right)-\sum_{i=\gamma+1}^{\infty}W\left(p^{i}\right)V_{i-1}\left(np^{-\gamma}\right)\right)+
\]
\[
+V_{\gamma-1}\left(np^{-\gamma}+ap^{-\gamma}\right)W\left(p^{\gamma}\right)\varOmega\left(|y-np^{-\gamma}-ap^{-\gamma}|_{p}p^{-\gamma+1}\right).
\]
Consequently,
\[
W_{m\left(x\right)}\varOmega\left(|y-np^{-\gamma}-ap^{-\gamma}|_{p}p^{-\gamma+1}\right)=
\]
\[
=W\left(|x-np^{-\gamma}-ap^{-\gamma}|_{p}\right)V_{\gamma-1}\left(np^{-\gamma}+ap^{-\gamma}\right)\left[1-\varOmega\left(|y-np^{-\gamma}-ap^{-\gamma}|_{p}p^{-\gamma+1}\right)\right]-
\]
\[
-\left(\sum_{i=\gamma}^{\infty}W\left(p^{i}\right)V_{i}\left(np^{-\gamma}\right)-\sum_{i=\gamma+1}^{\infty}W\left(p^{i}\right)V_{i-1}\left(np^{-\gamma}\right)\right)\varOmega\left(|y-np^{-\gamma}-ap^{-\gamma}|_{p}p^{-\gamma+1}\right)+
\]
\begin{equation}
+V_{\gamma-1}\left(np^{-\gamma}+ap^{-\gamma}\right)W\left(p^{\gamma}\right)\varOmega\left(|y-np^{-\gamma}-ap^{-\gamma}|_{p}p^{-\gamma+1}\right)\label{D_Omega}
\end{equation}
Next, we introduce the function
\begin{equation} \label{g_g_n_a_b}
\begin{gathered}
g_{\gamma,n,a,b}\left(x\right)=V_{\gamma-1}\left(np^{-\gamma}+bp^{-\gamma}\right)\varOmega\left(|x-np^{-\gamma}-ap^{-\gamma}|_{p}p^{-\gamma+1}\right)-\\
-V_{\gamma-1}\left(np^{-\gamma}+ap^{-\gamma}\right)\varOmega\left(|x-np^{-\gamma}-bp^{-\gamma}|_{p}p^{-\gamma+1}\right)
\end{gathered}
\end{equation}
where $\gamma\in\mathbb{Z}$, $n\in\mathbb{Q}_{p}/Z_{p}$ and $a,b=0,1,\ldots,p-1$,
and act by operator~(\ref{Vlad_Q_p_gen}) on~(\ref{g_g_n_a_b}),
using~(\ref{D_Omega}). So, if  $|x-np^{-\gamma}-ap^{-\gamma}|_{p}>p^{\gamma-1}$
and $|x-np^{-\gamma}-bp^{-\gamma}|_{p}>p^{\gamma-1}$, then
\[
W_{m\left(x\right)}\left[V_{\gamma-1}\left(np^{-\gamma}+bp^{-\gamma}\right)\varOmega\left(|x-np^{-\gamma}-ap^{-\gamma}|_{p}p^{-\gamma+1}\right)-\right.
\]
\[
-\left.V_{\gamma-1}\left(np^{-\gamma}+ap^{-\gamma}\right)\varOmega\left(|x-np^{-\gamma}-bp^{-\gamma}|_{p}p^{-\gamma+1}\right)\right]=0.
\]
Next, if $|x-np^{-\gamma}-ap^{-\gamma}|_{p}\leq p^{\gamma-1}$, then
\[
W_{m\left(x\right)}\left[V_{\gamma-1}\left(np^{-\gamma}+bp^{-\gamma}\right)\varOmega\left(|x-np^{-\gamma}-ap^{-\gamma}|_{p}p^{-\gamma+1}\right)-\right.
\]
\[
-\left.V_{\gamma-1}\left(np^{-\gamma}+ap^{-\gamma}\right)\varOmega\left(|x-np^{-\gamma}-bp^{-\gamma}|_{p}p^{-\gamma+1}\right)\right]=
\]
\[
=-V_{\gamma-1}\left(np^{-\gamma}+bp^{-\gamma}\right)\left(\sum_{i=\gamma}^{\infty}W\left(p^{i}\right)V_{i}\left(np^{-\gamma}\right)-\sum_{i=\gamma+1}^{\infty}W\left(p^{i}\right)V_{i-1}\left(np^{-\gamma}\right)\right)+
\]
\[
+V_{\gamma-1}\left(np^{-\gamma}+ap^{-\gamma}\right)V_{\gamma-1}\left(np^{-\gamma}+bp^{-\gamma}\right)W\left(p^{\gamma}\right)V_{\gamma-1}\left(np^{-\gamma}+ap^{-\gamma}\right)-
\]
\[
-V_{\gamma-1}\left(np^{-\gamma}+ap^{-\gamma}\right)V_{\gamma-1}\left(np^{-\gamma}+bp^{-\gamma}\right)W\left(|x-np^{-\gamma}-bp^{-\gamma}|_{p}\right)=
\]
\[
=-V_{\gamma-1}\left(np^{-\gamma}+bp^{-\gamma}\right)\left(\sum_{i=\gamma}^{\infty}W\left(p^{i}\right)V_{i}\left(np^{-\gamma}\right)-\sum_{i=\gamma+1}^{\infty}W\left(p^{i}\right)V_{i-1}\left(np^{-\gamma}\right)\right)=
\]
\[
=-V_{\gamma-1}\left(np^{-\gamma}+bp^{-\gamma}\right)\sum_{i=\gamma}^{\infty}\left(W\left(p^{i}\right)-W\left(p^{i+1}\right)\right)V_{i}\left(np^{-\gamma}\right)=
\]
\[
=\lambda_{\gamma,n}V_{\gamma-1}\left(np^{-\gamma}+bp^{-\gamma}\right).
\]
The same result, but with opposite sign and with the change $a\rightarrow b$, is obtained in the case
$|x-np^{-\gamma}-bp^{-\gamma}|_{p}\leq p^{\gamma-1}$.
This shows that functions (\ref{g_g_n_a_b}) are eigenfunctions of  operator~(\ref{Vlad_Q_p_gen}) with eigenvalues~(\ref{lambda_gn}).

Assume that $V_{\gamma-1}\left(np^{-\gamma}+ap^{-\gamma}\right)\neq0$,
then $V_{\gamma}\left(np^{-\gamma}\right)\neq0$. Consider the function
\[
\dfrac{1}{V_{\gamma}\left(np^{-\gamma}\right)}\sum_{b=0}^{p-1}g_{\gamma,n,a,b}\left(x\right)=
\]
\[
=\varOmega\left(|x-np^{-\gamma}-ap^{-\gamma}|_{p}p^{-\gamma+1}\right)-\dfrac{V_{\gamma-1}\left(np^{-\gamma}+ap^{-\gamma}\right)}{V_{\gamma}\left(np^{-\gamma}\right)}\varOmega\left(|x-np^{-\gamma}|_{p}p^{-\gamma}\right)=f_{\gamma,n,a}\left(x\right).
\]
By the construction, the functions $f_{\gamma,n,a}\left(x\right)$ are eigenfunctions of operator~(\ref{Vlad_Q_p_gen}) with eigenvalues
(\ref{lambda_gn}). This completes the proof of Theorem~1.
\end{proof}

\begin{theorem}
Functions \eqref{f_gna}, where $\gamma\in\mathbb{Z}$,
$n\in\mathbb{Q}_{p}/Z_{p}$, $a=0,1,\ldots,p-1$ are such that
$V_{\gamma-1}\left(np^{-\gamma}+ap^{-\gamma}\right)\neq0$, form a~compete system of vectors in the space
$L^{2}\left(\mathbb{Q}_{p},m\left(x\right)d_{p}x\right)$.
\label{th2}
\end{theorem}

\begin{proof}
If the condition $V_{\gamma-1}\left(np^{-\gamma}+ap^{-\gamma}\right)=0$ is satisfied with some $\gamma\in\mathbb{Z},\: n\in\mathbb{Q}_{p}/Z_{p}$, $ a=0,1,\ldots,p-1$,
then the ball $B_{\gamma-1}\left(np^{-\gamma}+ap^{-\gamma}\right)$ is a~nullset. The set of all supports of all balls of nonzero measure $m\left(x\right)d_{p}x$
forms a~basis for $L^{2}\left(\mathbb{Q}_{p},m\left(x\right)d_{p}x\right)$. Hence, it suffices to show that the characteristic function $\varOmega\left(|x-np^{-\gamma}-ap^{-\gamma}|_{p}p^{-\gamma+1}\right)$
of any ball  $B_{\gamma-1}\left(np^{-\gamma}+ap^{-\gamma}\right)$ from  $\mathbb{Q}_{p}$  of nonzero measure,
$\gamma\in\mathbb{Z},\: n\in\mathbb{Q}_{p}/Z_{p}$, $ a=0,1,\ldots,p-1$, can be expanded in functions~(\ref{f_gna}).

Let $\left\{ x\right\} $ and $\left[x\right]$ be, respectively, the fractional and integer parts of a number $x\in\mathbb{Q}_{p}$. Setting $n_{i}=\left\{ np^{i}\right\} $,
$i=0,1,2,\ldots$, we have
\begin{equation}
n_{i}p^{-\gamma+1}=\left\{ n_{i}p\right\} p^{-\gamma}+\left[n_{i}p\right]p^{-\gamma}=n_{i+1}p^{-\gamma}+\left[n_{i}p\right]p^{-\gamma}\label{n_i_p^gamma}
\end{equation}
and hence the function
\begin{gather*}
\varOmega\left(|x-n_{i}p^{-\gamma+1}|_{p}p^{-\gamma+1}\right)-\dfrac{V_{\gamma-1}\left(np^{-\gamma}+ap^{-\gamma}\right)}{V_{\gamma}\left(np^{-\gamma}\right)}\varOmega\left(|x-n_{i+1}p^{-\gamma}|_{p}p^{-\gamma}\right)=
\\
=\varOmega\left(|x-n_{i+1}p^{-\gamma}-\left[n_{i}p\right]p^{-\gamma}|_{p}p^{-\gamma+1}\right)-\dfrac{V_{\gamma-1}\left(n_{i+1}p^{-\gamma}+\left[n_{i}p\right]p^{-\gamma}\right)}{V_{\gamma}\left(n_{i+1}p^{-\gamma}\right)}\varOmega\left(|x-n_{i+1}p^{-\gamma}|_{p}p^{-\gamma}\right)
\end{gather*}
(see~(\ref{f_gna})) belongs to the complete system of functions~(\ref{f_gna}).

Assume that $V_{\gamma-1}\left(np^{-\gamma}+ap^{-\gamma}\right)\neq0$.
Since $m\left(x\right)$ is positive definite, we have $V_{\gamma+i}\left(n_{i}p^{-\gamma-i}\right)\leq V_{\gamma+i+1}\left(n_{i+1}p^{-\gamma-i+1}\right)$
for $i=0,1,\ldots$, and so  $V_{\gamma+i}\left(n_{i}p^{-\gamma-i}\right)\neq0$.
This allows us to consider the function $g_{\gamma,n,a}\left(x\right)$:

\begin{equation}
g_{\gamma,n,a}\left(x\right)=f_{\gamma,n,a}\left(x\right)+V_{\gamma-1}\left(np^{-\gamma}+ap^{-\gamma}\right)\sum_{i=1}^{\infty}\dfrac{1}{V_{\gamma+i-1}\left(n_{i-1}p^{-\gamma-i+1}\right)}f_{\gamma+i,n_{i},\left[n_{i-1}p\right]}\left(x\right).\label{g(x)-1}
\end{equation}
Next, consider the squared $L^{2}\left(\mathbb{Q}_{p},m\left(x\right)d_{p}x\right)$-norm,
\[
\left\Vert g_{\gamma,n,a}\left(x\right)-\varOmega\left(|x-np^{-\gamma}-ap^{-\gamma}|_{p}p^{-\gamma+1}\right)\right\Vert ^{2}=\intop_{\mathbb{Q}_{p}}m\left(x\right)d_{p}xg_{\gamma,n,a}^{2}\left(x\right)-
\]
\begin{equation}
-2\intop_{\mathbb{Q}_{p}}m\left(x\right)g_{\gamma,n,a}\left(x\right)\varOmega\left(|x-np^{-\gamma}-ap^{-\gamma}|_{p}p^{-\gamma+1}\right)+\intop_{\mathbb{Q}_{p}}m\left(x\right)\varOmega\left(|x-np^{-\gamma}-ap^{-\gamma}|_{p}p^{-\gamma+1}\right)
\label{norm_g-Omega}
\end{equation}
and consider the first integral in~(\ref{norm_g-Omega}):
\[
\intop_{\mathbb{Q}_{p}}m\left(x\right)d_{p}xg_{\gamma,n,a}^{2}\left(x\right)=\intop_{\mathbb{Q}_{p}}m\left(x\right)d_{p}xf_{\gamma,n,a}^{2}\left(x\right)+
\]

\[
+2\intop_{\mathbb{Q}_{p}}m\left(x\right)d_{p}xf_{\gamma,n,a}\left(x\right)V_{\gamma-1}\left(np^{-\gamma}+ap^{-\gamma}\right)\sum_{i=1}^{\infty}\dfrac{1}{V_{\gamma+i-1}\left(n_{i-1}p^{-\gamma-i+1}\right)}f_{\gamma+i,n_{i},\left[n_{i-1}p\right]}\left(x\right)+
\]

\[
+\intop_{\mathbb{Q}_{p}}m\left(x\right)d_{p}xf_{\gamma,n,a}\left(x\right)V_{\gamma-1}\left(np^{-\gamma}+ap^{-\gamma}\right)\sum_{i=1}^{\infty}\dfrac{1}{V_{\gamma+i-1}\left(n_{i-1}p^{-\gamma-i+1}\right)}f_{\gamma+i,n_{i},\left[n_{i-1}p\right]}\left(x\right)\times
\]
\[
\times V_{\gamma-1}\left(np^{-\gamma}+ap^{-\gamma}\right)\sum_{j=1}^{\infty}\dfrac{1}{V_{\gamma+j-1}\left(n_{j-1}p^{-\gamma-j+1}\right)}f_{\gamma+j,n_{j},\left[n_{j-1}p\right]}\left(x\right).
\]
From (\ref{restr1}) it follows by (\ref{f_gna}) that each of the sums is uniformly convergent to~$x$. Hence, we may
interchange the order of summation and integration. It is easily checked that
\[
\intop_{\mathbb{Q}_{p}}m\left(x\right)d_{p}xf_{\gamma,n,a}\left(x\right)f_{\gamma',n',a'}\left(x\right)=
\]
\begin{equation}
=\delta_{\gamma,\gamma'}\delta_{n,n'}\left(\delta_{a,a'}V^{\gamma-1}\left(np^{-\gamma}+ap^{-\gamma}\right)-\dfrac{V_{\gamma-1}\left(np^{-\gamma}+ap^{-\gamma}\right)V_{\gamma-1}\left(np^{-\gamma}+a'p^{-\gamma}\right)}{V_{\gamma}\left(np^{-\gamma}\right)}\right).\label{(f,f)}
\end{equation}
Using (\ref{(f,f)}), we see that
\[
\intop_{\mathbb{Q}_{p}}m\left(x\right)d_{p}xg_{\gamma,n,a}^{2}\left(x\right)=\left(V_{\gamma-1}\left(np^{-\gamma}+ap^{-\gamma}\right)-\dfrac{V_{\gamma-1}\left(np^{-\gamma}+ap^{-\gamma}\right)^{2}}{V_{\gamma}\left(np^{-\gamma}\right)}\right)+
\]
\[
+\sum_{i=1}^{\infty}V_{\gamma-1}\left(np^{-\gamma}+ap^{-\gamma}\right)^{2}\dfrac{1}{V_{\gamma+i-1}\left(n_{i-1}p^{-\gamma-i+1}\right)^{2}}\times
\]
\[
\times\left(V_{\gamma+i-1}\left(n_{i}p^{-\gamma-i}+\left[n_{i-1}p\right]p^{-\gamma-i}\right)-\dfrac{V_{\gamma+i-1}\left(n_{i}p^{-\gamma-i}+\left[n_{i-1}p\right]p^{-\gamma-i}\right)^{2}}{V_{\gamma+i}\left(n_{i}p^{-\gamma-i}\right)}\right).
\]
Next, by (\ref{n_i_p^gamma})
\[
n_{i}p^{-\gamma-i}+\left[n_{i-1}p\right]p^{-\gamma-i}=n_{i-1}p^{-\gamma-i+1},
\]
and hence,
\[
\intop_{\mathbb{Q}_{p}}m\left(x\right)d_{p}xg_{\gamma,n,a}^{2}\left(x\right)=\left(V_{\gamma-1}\left(np^{-\gamma}+ap^{-\gamma}\right)-\dfrac{V_{\gamma-1}\left(np^{-\gamma}+ap^{-\gamma}\right)^{2}}{V_{\gamma}\left(np^{-\gamma}\right)}\right)+
\]
\[
+\sum_{i=1}^{\infty}V_{\gamma-1}\left(np^{-\gamma}+ap^{-\gamma}\right)^{2}\dfrac{1}{V_{\gamma+i-1}\left(n_{i-1}p^{-\gamma-i+1}\right)}\times
\]
\[
\times\left(1-\dfrac{V_{\gamma+i-1}\left(n_{i-1}p^{-\gamma-i+1}\right)}{V_{\gamma+i}\left(n_{i}p^{-\gamma-i}\right)}\right)=
\]
\[
=V_{\gamma-1}\left(np^{-\gamma}+ap^{-\gamma}\right)^{2}\left(\dfrac{1}{V_{\gamma-1}\left(np^{-\gamma}+ap^{-\gamma}\right)}-\dfrac{1}{V_{\gamma}\left(np^{-\gamma}\right)}\right)+
\]
\[
+\sum_{i=1}^{\infty}V_{\gamma-1}\left(np^{-\gamma}+ap^{-\gamma}\right)^{2}\left(\dfrac{1}{V_{\gamma+i-1}\left(n_{i-1}p^{-\gamma-i+1}\right)}-\dfrac{1}{V_{\gamma+i}\left(n_{i}p^{-\gamma-i}\right)}\right)=
\]
\[
=V_{\gamma-1}\left(np^{-\gamma}+ap^{-\gamma}\right)
\]
A similar analysis shows that
\[
\intop_{\mathbb{Q}_{p}}m\left(x\right)d_{p}xg_{\gamma,n,a}\left(x\right)\varOmega\left(|x-np^{-\gamma}-ap^{-\gamma}|_{p}p^{-\gamma+1}\right)=V_{\gamma-1}\left(np^{-\gamma}+ap^{-\gamma}\right),
\]
\[
\intop_{\mathbb{Q}_{p}}m\left(x\right)d_{p}x\varOmega\left(|x-np^{-\gamma}-ap^{-\gamma}|_{p}p^{-\gamma+1}\right)=V_{\gamma-1}\left(np^{-\gamma}+ap^{-\gamma}\right).
\]
As a result,
\[
\left\Vert g_{\gamma,n,a}\left(x\right)-\varOmega\left(|x-np^{-\gamma}-ap^{-\gamma}|_{p}p^{-\gamma+1}\right)\right\Vert ^{2}=V_{\gamma-1}\left(np^{-\gamma}+ap^{-\gamma}\right)\left(1-2+1\right)=0.
\]
Since the norm of this expression is zero, we have
\[
\varOmega\left(|x-np^{-\gamma}-ap^{-\gamma}|_{p}p^{-\gamma+1}\right)=
\]
\[
=f_{\gamma,n,a}\left(x\right)+V_{\gamma-1}\left(np^{-\gamma}+ap^{-\gamma}\right)\sum_{i=1}^{\infty}\dfrac{1}{V_{\gamma+i-1}\left(n_{i-1}p^{-\gamma-i+1}\right)}f_{\gamma+i,n_{i},\left[n_{i-1}p\right]}\left(x\right),
\]
completing the proof of the theorem.
\end{proof}

\begin{remark} The complete system of functions $\left\{ f_{\gamma,n,a}\left(x\right)\right\} $
is overcomplete in $L^{2}\left(\mathbb{Q}_{p},m\left(x\right)d_{p}x\right)$,
because in a~family of~$p$ functions $f_{\gamma,n,a}\left(x\right)$ with $a=0,1,\ldots,p-1$ and fixed  $\gamma$ and~$n$
only $p-1$ functions are linearly independent. This can be verified from the following relation
\[
\sum_{a=0}^{p-1}f_{\gamma,n,a}\left(x\right)=0.
\]
\label{rem1}
\end{remark}

Using Theorems \ref{th1} and~\ref{th2} one may easily find the solution of the Cauchy problem of equation (\ref{V_Eq}) with initial condition
\[
f\left(x,0\right)=\Omega\left(|x|_{p}\right).
\]
From the expansion
\[
\varOmega\left(|x|_{p}\right)=\sum_{i=0}^{\infty}\dfrac{V_{0}\left(0\right)}{V_{i}\left(0\right)}f_{i+1,0,0}\left(x\right),
\]
we easily find that
\[
f\left(x,t\right)=\exp\left(W_{m\left(x\right)}t\right)\left(\sum_{i=0}^{\infty}\dfrac{V_{0}\left(0\right)}{V_{i}\left(0\right)}f_{i+1,0,0}\left(x\right)\right)=
\]
\begin{equation}
=\sum_{i=0}^{\infty}\exp\left(-\sum_{j=i+1}^{\infty}\left(W\left(p^{j}\right)-W\left(p^{j+1}\right)\right)V^{j}\left(0\right)t\right)\dfrac{V_{0}\left(0\right)}{V_{i}\left(0\right)}f_{i+1,0,0}\left(x\right).\label{solution_m(x)}
\end{equation}

\section{Orthonormal basis of real valued functions for $L^{2}\left(\mathbb{Q}_{p},m\left(x\right)d_{p}x\right)$}

The next theorem gives the construction of a orthonormal basis of real valued functions for   $L^{2}\left(\mathbb{Q}_{p},m\left(x\right)d_{p}x\right)$.

\begin{theorem}
The functions
\begin{equation}
\varphi_{\gamma,n,b}\left(x\right)=\dfrac{1}{k_{\gamma,n}}
\dfrac{\sqrt{V_{\gamma-1}\left(np^{-\gamma}+bp^{-\gamma}\right)}}{V_{\gamma-1}
\left(np^{-\gamma}\right)}\left(f_{\gamma,n,0}\left(x\right)+\dfrac{k_{\gamma,n}V_{\gamma-1}
\left(np^{-\gamma}\right)}{V_{\gamma-1}\left(np^{-\gamma}+bp^{-\gamma}\right)}f_{\gamma,n,b}
\left(x\right)\right),
\label{basis}
\end{equation}
where $\gamma\in\mathbb{Z},\: n\in\mathbb{Q}_{p}/Z_{p}$, $ b=1,\ldots,p-1$, $V_{\gamma-1}\left(np^{-\gamma}+bp^{-\gamma}\right)\neq0$,
and $k_{\gamma,n}=-1\pm\sqrt{\dfrac{V_{\gamma}\left(np^{-\gamma}\right)}{V_{\gamma-1}\left(np^{-\gamma}\right)}}$,
form a~orthonormal basis of real valued functions for  $L^{2}\left(B_{r},m\left(x\right)d_{p}x\right)$.
\label{th3}
\end{theorem}

\begin{proof}
It follows from Theorem~\ref{th2} and Remark~\ref{rem1}  that the system of functions (\ref{basis}) is complete. Let us prove that system (\ref{basis}) is orthonormal.
Consider the inner product $\left(\varphi_{\gamma,n,b}\left(x\right),\varphi_{\gamma',n',b'}\left(x\right)\right)$:

\[
\left(\varphi_{\gamma,n,b}\left(x\right),\varphi_{\gamma',n',b'}\left(x\right)\right)=\intop_{\mathbb{Q}_{p}}m\left(x\right)d_{p}x\varphi_{\gamma,n,b}\left(x\right)\varphi_{\gamma',n',b'}\left(x\right)=
\]

\[
=\dfrac{1}{k_{\gamma,n}k_{\gamma',n'}}\dfrac{\sqrt{V_{\gamma-1}\left(np^{-\gamma}+bp^{-\gamma}\right)V_{\gamma'-1}\left(n'p^{-\gamma'}+b'p^{-\gamma'}\right)}}{V_{\gamma-1}\left(np^{-\gamma}\right)V_{\gamma'-1}\left(n'p^{-\gamma'}\right)}\times
\]

\[
\times\intop_{\mathbb{Q}_{p}}m\left(x\right)d_{p}x\left[f_{\gamma,n,0}\left(x\right)+\dfrac{k_{\gamma,n}V_{\gamma-1}\left(np^{-\gamma}\right)}{V_{\gamma-1}\left(np^{-\gamma}+bp^{-\gamma}\right)}f_{\gamma,n,b}\left(x\right)\right]\times
\]

\[
\times\left[f_{\gamma',n',0}\left(x\right)+\dfrac{k_{\gamma',n'}V_{\gamma'-1}\left(n'p^{-\gamma'}\right)}{V_{\gamma'-1}\left(n'p^{-\gamma'}+b'p^{-\gamma'}\right)}f_{\gamma',n',b'}\left(x\right)\right]=
\]

\[
=\dfrac{1}{k_{\gamma,n}k_{\gamma',n'}}\dfrac{\sqrt{V_{\gamma-1}\left(np^{-\gamma}+bp^{-\gamma}\right)V_{\gamma'-1}\left(n'p^{-\gamma'}+b'p^{-\gamma'}\right)}}{V_{\gamma-1}\left(np^{-\gamma}\right)V_{\gamma'-1}\left(n'p^{-\gamma'}\right)}\times
\]

\[
\times\intop_{\mathbb{Q}_{p}}m\left(x\right)d_{p}x\left[f_{\gamma,n,0}\left(x\right)f_{\gamma,n,0}\left(x\right)+\dfrac{k_{\gamma,n}V_{\gamma-1}\left(np^{-\gamma}\right)}{V_{\gamma-1}\left(np^{-\gamma}+bp^{-\gamma}\right)}f_{\gamma,n,b}\left(x\right)f_{\gamma,n,0}\left(x\right)+\right.
\]

\[
+\dfrac{k_{\gamma',n'}V_{\gamma'-1}\left(n'p^{-\gamma'}\right)}{V_{\gamma'-1}\left(n'p^{-\gamma'}+b'p^{-\gamma'}\right)}f_{\gamma',n',b'}\left(x\right)f_{\gamma,n,0}\left(x\right)+
\]

\[
+\left.\dfrac{k_{\gamma,n}k_{\gamma',n'}V_{\gamma-1}\left(np^{-\gamma}\right)V_{\gamma'-1}\left(n'p^{-\gamma'}\right)}{V_{\gamma-1}\left(np^{-\gamma}+bp^{-\gamma}\right)V_{\gamma'-1}\left(n'p^{-\gamma'}+b'p^{-\gamma'}\right)}f_{\gamma,n,b}\left(x\right)f_{\gamma,n,b'}\left(x\right)\right].
\]
Using (\ref{(f,f)}) it follows that
\[
\intop_{\mathbb{Q}_{p}}m\left(x\right)d_{p}x\varphi_{\gamma,n,b}\left(x\right)\varphi_{\gamma',n',b'}\left(x\right)=
\]
\[
=\delta_{\gamma,\gamma'}\delta_{n,n'}\delta_{b,b'}+\delta_{\gamma,\gamma'}\delta_{n,n'}\dfrac{1}{k_{\gamma,n}^{2}}\dfrac{V_{\gamma-1}\left(np^{-\gamma}+bp^{-\gamma}\right)}{V_{\gamma-1}\left(np^{-\gamma}\right)^{2}}V_{\gamma-1}\left(np^{-\gamma}\right)\times
\]
\[
\times \left(1-\dfrac{V_{\gamma-1}\left(np^{-\gamma}\right)}{V_{\gamma}\left(np^{-\gamma}\right)}-\left(2k_{\gamma,n}+k_{\gamma,n}^{2}\right)\dfrac{V_{\gamma-1}\left(np^{-\gamma}\right)}{V_{\gamma}\left(np^{-\gamma}\right)}\right)=\delta_{\gamma,\gamma'}\delta_{n,n'}\delta_{b,b'},
\]
because $k_{\gamma,n}=-1\pm\sqrt{\dfrac{V_{\gamma}\left(np^{-\gamma}\right)}{V_{\gamma-1}\left(np^{-\gamma}\right)}}$
is a~root of the equation
\[
k^{2}\dfrac{V_{\gamma-1}\left(np^{-\gamma}\right)}{V_{\gamma}\left(np^{-\gamma}\right)}+2k\dfrac{V_{\gamma-1}\left(np^{-\gamma}\right)}{V_{\gamma}\left(np^{-\gamma}\right)}-\left(1-\dfrac{V_{\gamma-1}\left(np^{-\gamma}\right)}{V_{\gamma}\left(np^{-\gamma}\right)}\right)=0.
\]
This completes the proof of the theorem.
\end{proof}

\section{Conclusions}

The principal result of this paper is the method for describing stationary Markov processes on the class of ultrametric spaces~$\mathbb{U}$
isometrically embeddable into~$\mathbb{Q}_{p}$. This
method is capable of reducing  the study of such processes to the study of such processes in $\mathbb{Q}_{p}$, thereby enabling one
to employ the traditional methods of the $p$-adic mathematical physics for the study of such processes.
More specifically, we show that the problem of solving the Cauchy problem
for the Kolmogorov--Feller equation of a~stationary Markov
process on~$\mathbb{U}$ can be reduced to solving the Cauchy problem for a~pseudo-differential
equation on $\mathbb{Q}_{p}$ with the pseudo-differential operator~(\ref{W}) and with  non-translation-invariant measure.
Under this approach all the characteristics of such a~process can be expressed in terms of this solution. The method proposed above depends heavily on
the existence of an isometrical isomorphism
between~$\mathbb{U}$ and a~measurable subset
$M\subset\mathbb{Q}_{p}$. Using such a~mapping one may naturally define the measure on~$\mathbb{U}$ as the restriction to~$M$ of the Haar measure
on~$\mathbb{Q}_{p}$. Moreover, we find the spectrum of the pseudo-differential
operator~(\ref{W}) in the space $L^{2}\left(\mathbb{Q}_{p},m\left(x\right)d_{p}x\right)$
and give an explicit construction of a orthonormal basis of real valued functions for
$L^{2}\left(\mathbb{Q}_{p},m\left(x\right)d_{p}x\right)$ formed from the eigenfunctions of operator~(\ref{W}).

\end{document}